\documentclass{amsart}

\usepackage{enumitem} 

\usepackage{hyperref}

\newcommand{\gdel}[1]{}
\newcommand{\anadd}[1]{\textcolor{blue!60!black}{#1}}

\usepackage{tikz}
\usepackage{amsmath}
\usepackage{amssymb}
\usepackage{amsthm}
\usepackage{mathrsfs}

\usetikzlibrary{math} 
\usetikzlibrary{decorations.pathreplacing}

\newtheorem{prop}{Proposition}
\newtheorem{lemma}{Lemma}

\newcommand{\ie}{i.e., }

\makeatletter 

\@ifundefined{addto}{\def\addto#1#2{%
    \ifx#1\@undefined
    \def#1{#2}%
    \else
    \ifx#1\relax
    \def#1{#2}%
    \else
        {\toks@\expandafter{#1#2}%
          \xdef#1{\the\toks@}}%
        \fi
        \fi
}}{}

\def\addto@every@math{%
  \expandafter\addto\csname \expandafter\ifx
  \csname mathoptions@on\endcsname\relax 
  check@mathfonts\else mathoptions@on\fi\endcsname
}

\def\active@def#1{%
  \begingroup\lccode`\~=`#1\relax\lowercase{\endgroup\def~}%
}

\def\fixmathspacing#1#2{%
  \addto@every@math{%
    \catcode`#1=12 \mathcode`#1="8000
    \active@def#1{#2}%
  }%
}

\makeatother 

\fixmathspacing{/}{\mathclose{}\mathchar"013D\mathopen{}}

\title[extendible spacetime without CTCs whose every extension has CTCs]{An extendible spacetime without closed timelike curves whose every extension contains closed timelike curves}

\author{H.\ Andr{\'e}ka}
\address{H.\ Andr{\'e}ka, HUN-REN Alfr\'ed R\'enyi Institute of Mathematics, Budapest, Hungary}
\email{andreka.hajnal@renyi.hu}

\author{J.\ Madar{\'a}sz}
\address{J.\ Madar\'asz, HUN-REN Alfr\'ed R\'enyi Institute of Mathematics, Budapest, Hungary}
\email{madarasz.judit@renyi.hu}

\author{J. Manchak}
\address{J. Manchak, Department of Logic and Philosophy of Science, University of California, Irvine, USA}
\email{jmanchak@uci.edu}

\author{I.\ N{\'e}meti}
\address{I.\ N{\'e}meti, HUN-REN Alfr\'ed R\'enyi Institute of Mathematics, Budapest, Hungary}
\email{nemeti.istvan@renyi.hu}

\author{G.\ Sz{\'e}kely}
\address{G.\ Sz\'ekely, HUN-REN Alfr\'ed R\'enyi Institute of Mathematics, Budapest, Hungary  \& University of Public Service, Budapest, Hungary}
\email{szekely.gergely@renyi.hu}

\date{\today}

\begin{document}

\begin{abstract}
By removing a fractal from time-rolled Minkowski spacetime, we construct an extendible spacetime without closed timelike curves whose every extension contains closed timelike curves. This settles a question posed by Geroch. 
\end{abstract}

\maketitle

\section{Introduction}%

Geroch~\cite{G} has emphasized that a suitable definition of a  ``singularity" within the context of general relativity depends crutially on the definition of spacetime maximality. He has also suggested that the latter definition is far from clear given its sensitivity to a background collection of spacetimes. In recent years, there has been significant interest in better understanding this sensitivity  by exploring spacetime maximality outside the standard context, e.g., $C^0$-maximality \cite{GL, S0, S}. In his paper, Geroch articulated a number of ``important and unsolved problems" concerning spacetime maximality -- \ some of which remain open 50+ years later. The present work concerns one such problem.

Let $\mathscr{U}$ be the collection of all spacetimes $(M,g)$ where $M$ is a smooth, connected, Hausdorff manifold and $g$ is a smooth Lorentzian metric on $M$. For any collection $\mathscr{P}\subseteq \mathscr{U}$, let us say that a spacetime in $\mathscr{P}$ is $\mathscr{P}$-maximal if it is not isometric to a proper subset of another spacetime $(M',g')$ in $\mathscr{P}$. We say a spacetime in $\mathscr{P}$ is $\mathscr{P}$-extendible if it is not $\mathscr{P}$-maximal. Now consider the following condition on a collection $\mathscr{P}\subset \mathscr{U}$:\\

\begin{enumerate}
\item[(*)] Every $\mathscr{P}$-maximal spacetime is $\mathscr{U}$-maximal.\\
\end{enumerate} 

Many spacetime properties of interest do not satisfy (*). For instance, if $\mathscr{P}\subset {U}$ is the collection of all globally hyperbolic spacetimes, one can easily find a $\mathscr{P}$-maximal spacetime that is $\mathscr{U}$-extendible: the $t<0$ region of Misner spacetime is one such example \cite{Mi, CI}. This particular example can be used to obtain similar results for a number of other steps on the causal ladder: stable causality, strong causality, and (past and future) distinguishability. The case of causality is handled with another simple example \cite{M}. Consider any $\mathscr{U}$-maximal spacetime $(M,g)$ with a single closed null curve and then remove a point $p$ on this curve. The resulting spacetime $(M-\{p\}, g)$ is causal but it has only one proper extension: the acausal spacetime $(M,g)$. This last step relies on an intuitive but non-trivial result recently proved by Sbierski \cite{S}: If $(M,g)$ is a $\mathscr{U}$-maximal spacetime and $p \in M$, then the only proper $\mathscr{U}$-extension (up to isometry) of $(M-\{p\}, g)$  is $(M,g)$ itself. 

It can also be seen that some spacetime properties trivially satisfy (*): if $\mathscr{P}\subset {U}$ is the collection of all geodesically complete spacetimes, then it is immediate that each spacetime in $\mathscr{P}$ is $\mathscr{U}$-maximal and thus $\mathscr{P}$-maximal \cite{C, BEE}. But the status of (*) with respect to other spacetime properties is sometimes difficult to settle. Geroch ~\cite{G} wondered about four spacetime properties in particular:

\begin{enumerate}
\item[(i)] ``is a source-free solution to Einstein's equations"
\item[(ii)] ``has no closed timelike curves"
\item[(iii)] ``satisfies an energy condition"
\item[(iv)] ``has a Killing vector"
\end{enumerate}

Questions concerning (i) and (iv) are still open as far as we know. Question (iii) was settled by Manchak \cite{M} who showed that (*) is not satisfied by any of the standard energy conditions.

Question (ii) is settled in the present paper. Geroch had special interest in this case, he writes in \cite[p. 277]{G}: 
``In fact, the status of closed timelike curves with respect to condition [(*)]
would have an important bearing on whether or not the program for 
refining the notion of a ``singular point" by using extensions can be 
carried out for the ideal points construction." For some time there was hope for a positive resolution, but 
it turns out that condition (*) is not satisfied by the collection $\mathscr{P}\subset \mathscr{U}$ 
of all spacetimes without closed timelike curves (CTCs). 

In this paper, an example is presented which is a $\mathscr{U}$-extendible spacetime without CTCs whose every extension contains CTCs. The example is a variation of an idea that is often used in the foundations of Lorentzian causality theory: take Minkowski spacetime and ``roll it up" along the time direction and then remove points in an appropriate way \cite{HE, Min}. In our case, the set of removed points needs to be quite intricate; here it is a fractal and at the same time a generalized Cantor set. We note that because of this construction, the example $(M,g)$ should not be considered a ``time machine" since all of the CTCs in every extension $(M',g')$ are confined to the timelike past of $M$ in $M'$ \cite{K1, K2}.

\section{Puncturing time-rolled Minkowski spacetime to ban all CTCs}
\label{sec-B}

In this section we construct the desired spacetime for dimension 2, it will be extended to all dimensions in Section~\ref{sec-h}. 
Time-rolled Minkowski spacetime has plenty of CTCs, we construct a ``barrier" $B$ such that each CTC in time-rolled Minkowski spacetime has at least one point from $B$. Because all extensions of the example need to have CTCs, the set $B$ cannot contain any line segment.
 
The idea is to build an infinite horizontal wall from fractal bricks, each
of which is constructed as the Cantor set; but instead of removing
middle thirds starting from the interval $[0,1]$, we remove sets from
a closed trapezoid having lightlike legs. We do this such that what
remains are three disjoint closed trapezoids each similar to the
original one and whose corresponding legs are on the same lightlike
line. Then we continue this procedure with the small trapezoids ad infinitum. See Figure~\ref{fig-B}.

\begin{figure}
  \begin{center}
    
    \begin{tikzpicture}[scale=0.11]
      \tikzstyle{lumo}=[red,thick]
      \tikzstyle{peco}=[thick,black!50,fill=black!42]
 
      \tikzmath{\x=4.0; \l =2.4; \lx=\l*\x; \m=(3-\l)*\x/2/\l; \lm=\l*\m; \y=\x-2*\lm; \ly=\l*\y;  \llm=\l*\lm;\llx=\l*\lx;\lly=\l*\ly;  \lllm=\l*\llm; \lllx=\l*\llx; \llly=\l*\lly; \gap=\llx+\lly;\t=\llx-\llm;} 

      \newcommand{\barpecetaro}[2]{
        \begin{scope}[shift={(#1,#2)}]
          \draw[peco] (0,0) to (\x,0) to (\x-\m,\m) to (\m,\m)to cycle;
          \draw[peco,shift={(\x+\y,0)}] (0,0) to (\x,0) to (\x-\m,\m) to (\m,\m)to cycle;
          \draw[peco] (2*\x-\lm,\lm) to (\x-\lm,\lm) to (\x-\lm+\m,\lm-\m) to (2*\x-\lm-\m,\lm-\m) to cycle;
        \end{scope}
      }

      \newcommand{\rbarpecetaro}[2]{
        \begin{scope}[shift={(#1,#2)}]
          \draw[peco] (0,0) to (\x,0) to (\x-\m,-\m) to (\m,-\m)to cycle;
          \draw[peco,shift={(\x+\y,0)}] (0,0) to (\x,0) to (\x-\m,-\m) to (\m,-\m)to cycle;
          \draw[peco] (2*\x-\lm,-\lm) to (\x-\lm,-\lm) to (\x-\lm+\m,-\lm+\m) to (2*\x-\lm-\m,-\lm+\m) to cycle;
        \end{scope}
      }


      \begin{scope}[shift={(0,1.5*\lllm)}]
        \draw[peco,shift={(-\gap-\t,0)}] (0,0)  node[black,above left]  {$B^0_0$} to  (\llx,0) to (\llx-\llm,\llm) to (\llm,\llm)  to cycle;
        \draw[peco, shift={(-\t,\llm)}, cm={1,0,0,-1,(0,0)}] (0,0) to node[above,black]  {$B^0_{-1}$}  (\llx,0) to (\llx-\llm,\llm) to (\llm,\llm)to cycle;
        \draw[peco] (0,0) to  (\llx,0) to  (\llx-\llm,\llm) to  node[above,black]  {$B^0_0$}  (\llm,\llm)to cycle;
        \draw[peco, shift={(\t,\llm)}, cm={1,0,0,-1,(0,0)}] (0,0) to node[above,black]  {$B^0_{1}$}  (\llx,0) to (\llx-\llm,\llm) to (\llm,\llm)to cycle;
        \node[gray] at (-\llx+2,\llm/3) {$\ldots$};
        \node[gray] at (\llx+\t+1,\llm/3) {$\ldots$};
        \node[right] at (\llx+\t+3,\llm/2) {$B^0$};
      \end{scope}


      \begin{scope}[shift={(-\t-\gap,0)}]
        \draw[peco] (0,0)  node[black,above left]  {$B^1_0$} to (\lx,0) to (\lx-\lm,\lm) to (\lm,\lm)to cycle;
        \draw[peco,shift={(\ly+\lx,0)}] (0,0) to (\lx,0) to (\lx-\lm,\lm) to (\lm,\lm)to cycle;
        \draw[peco,shift={(\lx-\llm,\llm)}] (0,0) to (\lx,0) to (\lx-\lm,-\lm) to (\lm,-\lm)to cycle;
\end{scope}

      \begin{scope}[shift={(-\t,\llm)}, cm={1,0,0,-1,(0,0)}]
        \node[above,black] at (\llx/2,0) {$B^1_{-1}$};
        \draw[peco] (0,0) to (\lx,0) to (\lx-\lm,\lm) to (\lm,\lm)to cycle;
        \draw[peco,shift={(\ly+\lx,0)}] (0,0) to (\lx,0) to (\lx-\lm,\lm) to (\lm,\lm)to cycle;
        \draw[peco,shift={(\lx-\llm,\llm)}] (0,0) to (\lx,0) to (\lx-\lm,-\lm) to (\lm,-\lm)to cycle;
      \end{scope}

      \node[gray] at (-\llx+2,\llm/3) {$\ldots$};
      
      \node[above,black] at (\llx/2,\llm) {$B^1_{0}$};
      \draw[peco] (0,0) to (\lx,0) to (\lx-\lm,\lm) to (\lm,\lm)to cycle;
      \draw[peco,shift={(\lx-\llm,\llm)}] (0,0) to (\lx,0) to (\lx-\lm,-\lm) to (\lm,-\lm)to cycle;
      \draw[peco,shift={(\ly+\lx,0)}] (0,0) to (\lx,0) to (\lx-\lm,\lm) to (\lm,\lm)to cycle;

      \begin{scope}[shift={(\t,\llm)}, cm={1,0,0,-1,(0,0)}]
        \node[above,black] at (\llx/2,0) {$B^1_{1}$};
        \draw[peco] (0,0) to (\lx,0) to (\lx-\lm,\lm) to (\lm,\lm)to cycle;
        \draw[peco,shift={(\ly+\lx,0)}] (0,0) to (\lx,0) to (\lx-\lm,\lm) to (\lm,\lm)to cycle;
        \draw[peco,shift={(\lx-\llm,\llm)}] (0,0) to (\lx,0) to (\lx-\lm,-\lm) to (\lm,-\lm)to cycle;
        \node[gray] at (\llx+1,2*\llm/3) {$\ldots$};
        \node[right] at (\llx+3,\llm/2) {$B^1$};
      \end{scope}

      \begin{scope}[shift={(0,-1.5*\lllm)}]
        \begin{scope}[shift={(-\gap-\t,0)}]
          \node[above left] at (0,0) {$B^2_0$};
          \barpecetaro{0}{0};
          \rbarpecetaro{\lx-\llm}{\llm};
          \barpecetaro{\lx+\ly}{0};
        \end{scope}

        \node[gray] at (-\llx+2,\llm/3) {$\ldots$};
        
        \begin{scope}[shift={(-\t,\llm)}, cm={1,0,0,-1,(0,0)}]
          \node[above,black] at (\llx/2,0) {$B^2_{-1}$};
          \barpecetaro{0}{0};
          \rbarpecetaro{\lx-\llm}{\llm};
          \barpecetaro{\lx+\ly}{0};
        \end{scope}

        \begin{scope}[shift={(0,0)}]
          \node[above,black] at (\llx/2,\llm) {$B^2_{0}$};
          \barpecetaro{0}{0};
          \rbarpecetaro{\lx-\llm}{\llm};
          \barpecetaro{\lx+\ly}{0};
        \end{scope}

        \begin{scope}[shift={(\t,\llm)}, cm={1,0,0,-1,(0,0)}]
          \node[above,black] at (\llx/2,0) {$B^2_{1}$};
          \barpecetaro{0}{0};
          \rbarpecetaro{\lx-\llm}{\llm};
          \barpecetaro{\lx+\ly}{0};
          \node[gray] at (\llx+1,2*\llm/3) {$\ldots$};
          \node[right] at (\llx+3,\llm/2) {$B^2$};
        \end{scope}
      \end{scope}

      \draw[gray,loosely dashed,very thick] (-\gap/2+\llm/2-\t/2,2.5*\lllm) to (-\gap/2+\llm/2-\t/2,-2.5*\lllm);

      \node[gray] at (-\gap-\t+\llx/2,-2*\lllm) {\Large $\vdots$};
      
      \node[gray] at (\llx/2,-2*\lllm) {\Large $\vdots$};

    \end{tikzpicture}
  \end{center}
   
  \caption{Construction of the barrier $B$ by an infinite iteration.\label{fig-B}}
\end{figure}
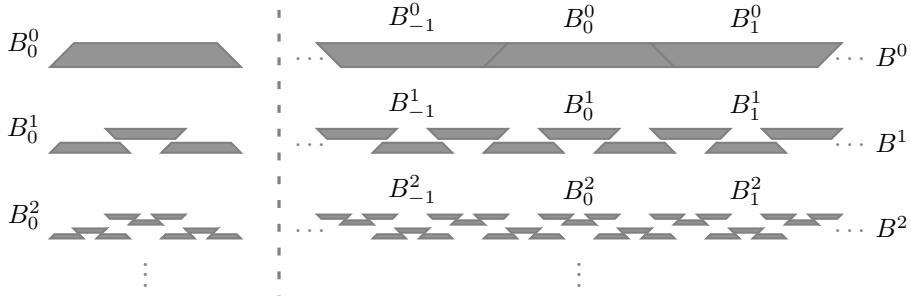

To see the condition for self-similarity, let $S$ be the length of the
longer base and $T$ be the height of the trapezoid, and let $\lambda$
be the ratio of similarity, see Figure~\ref{fig-selfsim}. Then,
because the legs of the trapezoids are lightlike, we have
$S=\frac{3}{\lambda}S-2T$, which can be reorganized as
\begin{equation}\label{eq-selfsim}
	2\lambda T=(3-\lambda)S.
\end{equation}
If $\lambda=3$, then $T=0$ and our construction degenerates into a
line segment. If $\lambda=2$, then $S=4T$ and the smaller trapezoids
overlap. However, any of the cases $2<\lambda<3$ are equally good for
our purposes.

To construct the fractal bricks, let us choose and fix some $\lambda$ for which $2<\lambda<3$. 
For concreteness, choose 
$T=1$, then $S=2\lambda\slash(3-\lambda)$.%
\footnote{$S=6$ if $\lambda=2.25$. The figures are constructed with $\lambda=2.4$.}
Let $B^0_0$ be the closed trapezoid with nodes $(0,0), (1,1), (1,S-1)$ and $(0,S)$. 
Then $B_0^0$ has horizontal bases and
lightlike legs such that its height and longer base satisfy equation
\eqref{eq-selfsim}. See Figure~\ref{fig-selfsim}.  

Let $B^1_0$ be the union of the three closed
smaller trapezoids similar to $B^0_0$ placed as in
Figure~\ref{fig-selfsim}, and let $B^{i+1}_0$ be the union of the
$3^{i+1}$ trapezoids that we get by iterating the same replacement
step with all the $3^i$ trapezoids of $B^i_0$. 
Let $m$ be the length
of the midsegment of trapezoid $B^0_0$. For each $n\in\mathbb{Z}$, let
$B^i_{n}$ be $B^i_0$ translated horizontally by $n m$ and rotated
around its center if $n$ is odd. See Figure~\ref{fig-B}.

We now define our barrier set  $B$ as
\begin{equation}\label{eq-csik}
	B^i := \bigcup_{n\in\mathbb{Z}}B^i_n\qquad\mbox{and}\qquad B=\bigcap_{i=0}^{\infty}B^i.
\end{equation}

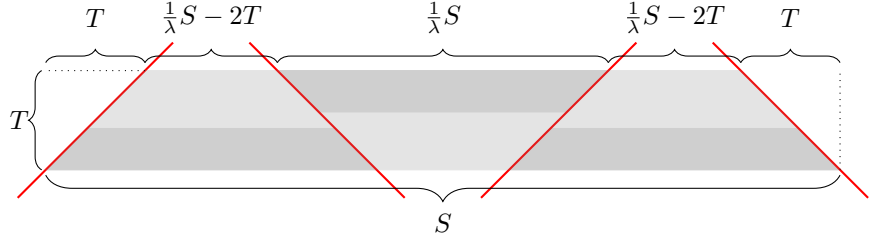
\begin{figure}
  \begin{center}
    \begin{tikzpicture}[scale=1.1]
      \tikzstyle{lumo}=[red,thick]
      \tikzstyle{koloro}=[gray,fill=gray,opacity=0.21]

      \tikzmath{\x=4.0; \l =2.4; \lx=\l*\x; \m=(3-\l)*\x/2/\l; \lm=\l*\m; \y=\x-2*\lm; \ly=\l*\y;  \llm=\l*\lm;\llx=\l*\lx;\lly=\l*\ly;  \lllm=\l*\llm; \lllx=\l*\llx; \llly=\l*\lly;} 

      \newcommand{\barpecetaro}[2]{
        \begin{scope}[shift={(#1,#2)}]
          \draw[koloro] (0,0) to (\x,0) to (\x-\m,\m) to (\m,\m)to cycle;
          \draw[koloro,shift={(\x+\y,0)}] (0,0) to (\x,0) to (\x-\m,\m) to (\m,\m)to cycle;
          \draw[koloro] (2*\x-\lm,\lm) to (\x-\lm,\lm) to (\x-\lm+\m,\lm-\m) to (2*\x-\lm-\m,\lm-\m) to cycle;
        \end{scope}
      }

      \barpecetaro{0}{0};
      \draw[lumo,shorten >= -15, shorten <=-15] (0,0) to (\lm,\lm);
      \draw[lumo,shorten >= -15, shorten <=-15] (\lx-\lm,\lm) to (\lx,0);
      \draw[lumo,shorten >= -15, shorten <=-15] (\x,0) to (\x-\lm,\lm);
      \draw[lumo,shorten >= -15, shorten <=-15] (\lx-\x,0) to (\lx-\x+\lm,\lm);

      \draw[dotted] (\lm,\lm) to (0,\lm);
      \draw[dotted] (\lx,0) to (\lx,\lm);

      \draw[koloro] (0,0) to (\lm,\lm) to (\lx-\lm,\lm) to (\lx,0) to cycle;

      \draw[decorate,decoration={brace,amplitude=10pt,mirror,raise=2}]
      (0,0) -- (\lx,0) node[midway,yshift=-20]{$S$};

      \draw[decorate,decoration={brace,amplitude=5pt,raise=1.5}]
      (0,0) -- (0,\lm) node[midway,xshift=-10]{$T$};

      \draw [decorate,decoration={brace,amplitude=5pt,raise=3}]
      (\x-\lm,\lm) -- (2*\x-\lm,\lm) node[midway,yshift=20]{$\frac{1}{\lambda}S$};

      \draw[decorate,decoration={brace,amplitude=5pt,raise=3}]
      (\lm,\lm) -- (\x-\lm,\lm) node[midway,yshift=20]{$\frac{1}{\lambda}S-2T$};

      \draw[decorate,decoration={brace,amplitude=5pt,raise=3}]
      (2*\x-\lm,\lm) -- (\lx-\lm,\lm) node[midway,yshift=20]{$\frac{1}{\lambda}S-2T$};

      \draw[decorate,decoration={brace,amplitude=5pt,raise=3}]
      (0,\lm) -- (\lm,\lm) node[midway,yshift=20]{$T$};

      \draw[decorate,decoration={brace,amplitude=5pt,raise=3}]
      (\lx-\lm,\lm) -- (\lx,\lm) node[midway,yshift=20]{$T$};

    \end{tikzpicture}
    \caption{The figure illustrates the main construction step of
      barrier $B$ together with the variables of equation
      \eqref{eq-selfsim} that is used to derive the condition
      $2<\lambda<3$ of self-similarity. The light gray part is the set
      that we remove from the trapezoid, and the three dark gray
      trapezoids are the remaining closed trapezoids which are similar
      to the original one.  Analogously to the construction of the
      Cantor set, we repeat this removal step ad infinitum.}
    \label{fig-selfsim}
  \end{center}
\end{figure}

We note that $B\subseteq [0,1]\times \mathbb{R}$ and
the horizontal projection of $B$ to the vertical interval $[(0,0),(1,0)]$
is almost the Cantor set. The only difference is that the lengths of the removed intervals
here are shorter. This is so because the lengths of the removed
intervals are $(1-2/\lambda)$-th of the original ones, which is less
than one third since $2<\lambda<3$. 
The intersection of $B$ with the horizontal line segment $[(0,0),(0,S)]$ is also almost the Cantor set.

\begin{prop}\label{prop-closed} ${}$
  \begin{enumerate}[label=(\alph*)]
  \item $B$ is closed.
  \item The complement of $B$ is connected.
  \end{enumerate}
\end{prop}

\begin{proof}
  It is easy to see that, for each $i$, the set
  $B^i=\bigcup_{n\in\mathbb{Z}} B_n^i$ is closed because it is the
  union of isolated-enough trapezoids. For a formal proof, $B^i$ is
  the union of the locally finite collection of closed sets $B_n^i$,
  \ie there is a neighborhood of each point in $\mathbb{R}^2$ such
  that it intersects only finitely many of the closed sets
  \cite[Cor.1.1.12]{Eng89}. Then $B$ is closed because it is an
  intersection of closed sets by \eqref{eq-csik}.

  For each $i\ge 1$, the complement of $B^i$ 
  is clearly connected since we just removed some isolated
  closed trapezoids from the plane. 
  The complement of $B$ is the union of the
  complements of  $B^i$. 
  Hence, the complement of $B$ is connected because it is the union of an upward directed collection of connected
  sets.
  %
\end{proof}

The following proposition says that no  causal curve can cross region $[0,1]\times\mathbb{R}$ 
without intersecting $B$. It is here where we use that the legs of the trapezoids are lightlike.

\begin{prop}\label{prop-blocking}
  Assume that $\gamma:[0,1]\to\mathbb{R}^2$ is a broken
  future-directed causal curve such that $\gamma(0)=(0,x)$ and
  $\gamma(1)=(1,y)$ for some $x,y\in\mathbb{R}$. Then there is $0\le
  t\le 1$ such that $\gamma(t)\in B$.
\end{prop}

\begin{proof}
  Let $\gamma$ be a causal curve in Minkowski spacetime $(\mathbb{R}^2,\eta)$ as in the
  statement. The plan of the proof is as follows. Using induction, we
  are going to show that there is an $n\in\mathbb{Z}$ such that, for
  all natural numbers $i$, $\gamma$ intersects one of the trapezoids
  of $B^i_n$ at its longer base. Moreover, if $\gamma$ intersects a
  trapezoid in $B_n^i$, then it intersects a sub-trapezoid in
  $B_n^{i+1}$.  The intersection of these nested closed trapezoids is
  a point $p$ of $B$ to which the intersection points converge. Thus
  $p$ has to be on $\gamma$ because the latter is continuous. This
  will complete the proof of the proposition.

  By definition, $B^0$ is $[0,1]\times\mathbb{R}$, so $\gamma(0)$ is
  on one of the bases of some $B_n^0$. If $\gamma(0)$ is on the
  shorter base of $B_n^0$, then $\gamma(1)$ is on the longer base of
  $B_n^0$, because the legs of $B_n^0$ are lightlike and $\gamma$ is
  broken future-directed causal. See Figure~\ref{figure-PF}. Thus
  either $\gamma(0)=(0,x)$ or $\gamma(1)=(1,y)$ is on the longer base
  of $B^0_n$ for some $n\in\mathbb{Z}$.  This shows the base case
  $i=0$ of the induction.

  Assume now that we have already seen that $\gamma$ intersects one of
  the trapezoids of $B^i_n$ at its longer base.  By symmetry, without
  loss of generality, we can assume that this longer base is at the
  bottom.  Then there are three cases: either $\gamma$ intersects the
  left closed, middle open or the right closed part of this bottom
  base. In the first and third cases, $\gamma$ intersects the longer
  base of the left or the right smaller trapezoid of the next
  iteration step.  In the second case,
  $\gamma$ has to intersect the longer base of the middle trapezoid
  because the common legs of these trapezoids are lightlike.  Hence, $\gamma$ intersects one of the
  included trapezoids of $B^{i+1}_n$ at its longer base, and this is
  what we wanted to show.
  
  Let now $p_i$ be the intersection points of $\gamma$ with the longer
  bases and let $p$ be the unique point in the intersection of the
  nested trapezoids $\gamma$ intersects. Then $p\in B$ by
  definition. Also, the $p_i$ converge to $p$ because each open
  neighbourhood $O$ of $p$ contains one of these nested trapezoids as
  a subset, since their diameters tend to $0$ as $i$ tends to
  infinity. But then $O$ contains all $p_j$, $j\ge i$ for some $i$. We
  have seen that the $p_i$ converge to $p$. Since $\gamma$ is
  continuous, then $\gamma(t)=p$ for some $0\le t\le 1$ because for
  all $i$, we have that $p_i=\gamma(t_i)$ for some $0\le t_i\le 1$.
\end{proof}

\begin{figure}  
	\begin{center}
		\begin{tikzpicture}[scale=0.2,cm={1,0,0,-1,(0,0)}]
			\tikzstyle{lumo}=[red,dashed]
			\tikzstyle{peco}=[black!33,fill=black!33]
			
			\tikzmath{\x=4.0; \l =2.4; \lx=\l*\x; \m=(3-\l)*\x/2/\l; \lm=\l*\m; \y=\x-2*\lm; \ly=\l*\y;  \llm=\l*\lm;\llx=\l*\lx;\lly=\l*\ly;  \lllm=\l*\llm; \lllx=\l*\llx; \llly=\l*\lly;} 
			
			\coordinate (b0) at (\lly+\llx+\lx,-\llm) ;
			\coordinate (a0) at (\lly+\llx+\lx,\lllm+\llm) ;
			\coordinate (a1) at  (\lly+\llx+\lx,0);
			\coordinate (a2) at  (\lly+\llx,0);
			\coordinate (a3) at  (\llx+\lly+\lllm,\lllm);
			\coordinate (a4) at  (\llx+\lly-\lx+\lllm,\lllm);
			\coordinate (a4) at  (\llx+\lly-\lx+\lllm,\lllm);
			\coordinate (a5) at  (\llx+\lly-\lx+\lllm+\llm,\lllm-\llm);
			\coordinate (a6) at  (\ly+\lly+\lllm+\llm,\lllm-\llm);
			\coordinate (a7) at  (\ly+\lly+\lllm+\llm+\llm,\lllm);
			\coordinate (a8) at  (\ly+\lly-\lx+\lllm+\llm+\llm,\lllm);
			\coordinate (a9) at  (\llx,0);
			\coordinate (a10) at  (\llx-\lx,0);
			\coordinate (a11) at  (\llx-\lx+\llm,\llm);
			\coordinate (a12) at  (\lx-\llm,\llm);
			\coordinate (a13) at  (\lx,0);
			\coordinate (a14) at (0,0) circle [radius=.1];
			\coordinate (a15) at  (\lllm,\lllm);
			\coordinate (a16) at  (\lllm-\lx,\lllm);
			\coordinate (a17) at  (\lllm-\lx+\llm,\lllm-\llm);
			\coordinate (a18) at  (\lllm-\lx+\llm-\lx,\lllm-\llm);
			\coordinate (a19) at  (\lllm-\llx+\lx,\lllm);
			\coordinate (a20) at  (\lllm-\llx,\lllm);
			\coordinate (a21) at  (-\lly,0);
			\coordinate (a22) at  (-\lx-\lly,0);
			\coordinate (a23) at  (-\lx-\lly,\lllm+\llm);
			\coordinate (b23) at  (-\lx-\lly,-\llm);
			
			\draw[yellow, fill=yellow]  (a0) to (a1)   to (a2) to (a3) to (a4) to (a5) to (a6) to (a7) to (a8) to (a9) to (a10) to (a11) to (a12) to (a13) to (a14) to (a15) to (a16) to (a17) to (a18) to (a19) to (a20) to (a21) to (a22) to (a23) to cycle;
			

			
			\draw[green, fill=green]  (b0) to (a1)   to (a2) to (a3) to (a4) to (a5) to (a6) to (a7) to (a8) to (a9) to (a10) to (a11) to (a12) to (a13) to (a14) to (a15) to (a16) to (a17) to (a18) to (a19) to (a20) to (a21) to (a22) to (b23) to cycle;
			

			
			\draw  (-\lx-\lly+.5,0)  to (-\lx-\lly-.5,0)   node[left]  {$1$};
			
			\draw  (-\lx-\lly+.5,\lllm)  to (-\lx-\lly-.5,\lllm)   node[left]  {$0$};

			\draw[peco] (0,0) to (\lx,0) to (\lx-\lm,\lm) to (\lm,\lm)to cycle;
			\draw[peco,shift={(\ly+\lx,0)}] (0,0) to (\lx,0) to (\lx-\lm,\lm) to (\lm,\lm)to cycle;
			\draw[peco,shift={(\lx-\llm,\llm)}] (0,0) to (\lx,0) to (\lx-\lm,-\lm) to (\lm,-\lm)to cycle;
			
			\begin{scope}[shift={(\llx-\lllm,\lllm)}]
				\draw[peco] (0,0) to (\lx,0) to (\lx-\lm,-\lm) to (\lm,-\lm)to cycle;
				\draw[peco,shift={(\ly+\lx,0)}] (0,0) to (\lx,0) to (\lx-\lm,-\lm) to (\lm,-\lm)to cycle;
				\draw[peco,shift={(\lx-\llm,-\llm)}] (0,0) to (\lx,0) to (\lx-\lm,\lm) to (\lm,\lm)to cycle;
			\end{scope}
			
			\begin{scope}[shift={(-\llx+\lllm,\lllm)}]
				\draw[peco] (0,0) to (\lx,0) to (\lx-\lm,-\lm) to (\lm,-\lm)to cycle;
				\draw[peco,shift={(\ly+\lx,0)}] (0,0) to (\lx,0) to (\lx-\lm,-\lm) to (\lm,-\lm)to cycle;
				\draw[peco,shift={(\lx-\llm,-\llm)}] (0,0) to (\lx,0) to (\lx-\lm,\lm) to (\lm,\lm)to cycle;
			\end{scope}

			\begin{scope}[shift={(\lly+\llx,0)}]
				\draw[peco] (0,0) to (\lx,0) to (\lx-\lm,\lm) to (\lm,\lm) to cycle;
			\end{scope}
			
			\begin{scope}[shift={(-\lly-\lx,0)}]
				\draw[peco] (0,0) to (\lx,0) to (\lx-\lm,\lm) to (\lm,\lm)to cycle;
			\end{scope}
			
			\draw[lumo] (-\lly-\lllm,\lllm) to (-\lly,0);
			\draw[lumo] (\llx+\lly+\lllm,\lllm) to (\llx+\lly,0);
			\draw[lumo] (0,0) to (\lllm,\lllm);
			\draw[lumo] (\llx,0) to (\llx-\lllm,\lllm);
			
		\end{tikzpicture}
		\caption{No causal curve can cross region $[0,1]\times\mathbb{R}$
			without intersecting $B$, cf.\ Proposition~\ref{prop-blocking}
			below.  \label{figure-PF}}
	\end{center}
\end{figure}
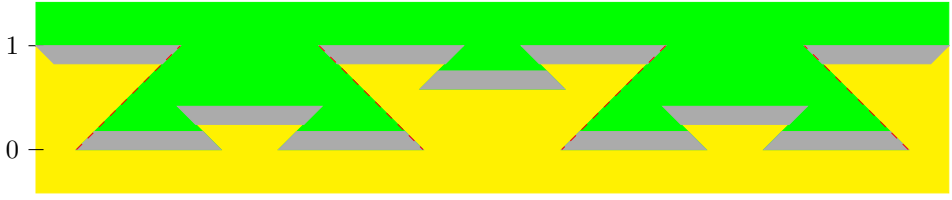

By Proposition~\ref{prop-closed}, we have that $(\mathbb{R}^2\setminus
B,\eta)$ is a spacetime. We can roll it up by choosing a large enough
$\ell\in\mathbb{Z}$ and gluing points $(-\ell,x)$ and $(\ell,x)$
together, for all $x\in\mathbb{R}$. For concreteness, choose an $\ell > \lambda\slash(3-\lambda)
= S/2$ (cf.\ equation~\eqref{eq-selfsim}) and let
$\mathbf{M^-}=(M^-,\eta)$ be the spacetime that we get this way. By
Proposition~\ref{prop-blocking}, there are no closed timelike curves
(CTCs) in $\mathbf{M^-}$. Clearly, $\mathbf{M^-}$ is extendible,
because, e.g., Minkowski spacetime rolled up at the same $\ell$ is an
extension of $\mathbf{M^-}$. We can call $\mathbf{M^-}$
\emph{punctured time-rolled Minkowski spacetime}.

In the following, we state three lemmas about $B$ that we will use in
the next section when proving that each proper extension of
$\mathbf{M^-}$ does have a CTC. In the rest of this section, we will
work in $(\mathbb{R}^2,\eta)$.
\bigskip

Because the diameters of the trapezoids tend to 0, every point of $B$
can be described by a \emph{choice sequence} that starts with an
integer and continues with an infinite series of decision of
$\mathcal{L}$eft, $\mathcal{M}$iddle, $\mathcal{R}$ight. In other
words, there is a one-to-one correspondence between points of $B$ and
$\mathbb{Z}\times\{\mathcal{L},\mathcal{M},\mathcal{R}\}^\omega$.

We call a point $e\in B$ \emph{eventually middle} if{}f its choice
sequence contains only finitely many $\mathcal{L}$ and $\mathcal{R}$
choices, in other words, it becomes constant $\mathcal{M}$ after some
time.

\begin{lemma}\label{lem-em-dense}
  The set of eventually middle points is everywhere
  dense in $B$, \ie for every $\varepsilon >0$ and for every $p\in B$
  there is an eventually middle point $e$ such that
  $|p-e|<\varepsilon$.
\end{lemma}

\begin{proof}
  This follows easily from the construction because every point \anadd{of $B$} is in a
  small enough closed trapezoid of the construction and every such
  trapezoid contains (infinitely many) eventually middle points.
\end{proof}

The statement of the next lemma is illustrated in Figure~\ref{fig-M0}.

\begin{lemma}\label{lem-CL}
  Let $e$ be an eventually middle point. Inside the smallest light
  rhombus $R$ containing the trapezoid after which we always choose
  $\mathcal{M}$ to reach $e$, the vertical centerline of $R$ intersects
  $B$ only in $e$.
\end{lemma}

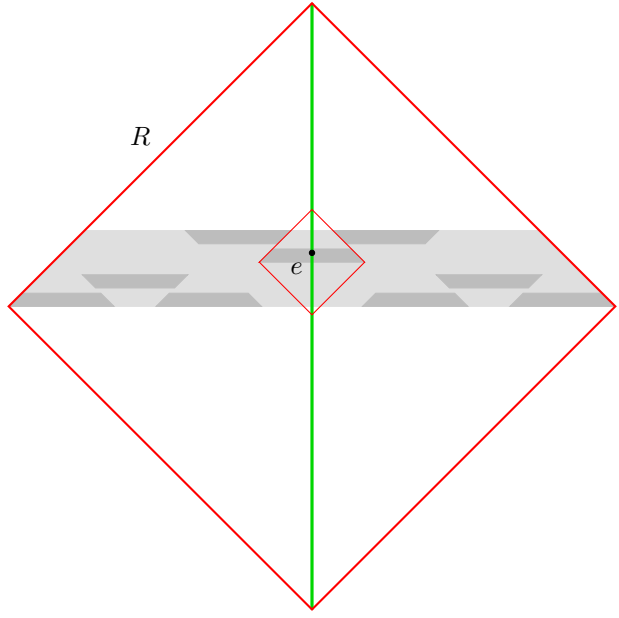
\begin{figure}
  \begin{tikzpicture}[scale=0.35]
    \tikzstyle{lumo}=[red,thick]
    \tikzstyle{peco}=[black!50,fill=black!50,opacity=0.35]
    \tikzstyle{verda}=[green!84!black] 

    \tikzmath{\x=4.0; \l =2.4; \lx=\l*\x; \m=(3-\l)*\x/2/\l; \lm=\l*\m; \y=\x-2*\lm; \ly=\l*\y;  \llm=\l*\lm;\llx=\l*\lx;\lly=\l*\ly;  \lllm=\l*\llm; \lllx=\l*\llx; \llly=\l*\lly;} 

    \newcommand{\barpecetaro}[2]{
      \begin{scope}[shift={(#1,#2)}]
        \draw[peco] (0,0) to (\x,0) to (\x-\m,\m) to (\m,\m)to cycle;
        \draw[peco,shift={(\x+\y,0)}] (0,0) to (\x,0) to (\x-\m,\m) to (\m,\m)to cycle;
        \draw[peco] (2*\x-\lm,\lm) to (\x-\lm,\lm) to (\x-\lm+\m,\lm-\m) to (2*\x-\lm-\m,\lm-\m) to cycle;
      \end{scope}
    }

    \newcommand{\rbarpecetaro}[2]{
      \begin{scope}[shift={(#1,#2)}]
        \draw[peco] (0,0) to (\x,0) to (\x-\m,-\m) to (\m,-\m)to cycle;
        \draw[peco,shift={(\x+\y,0)}] (0,0) to (\x,0) to (\x-\m,-\m) to (\m,-\m)to cycle;
        \draw[peco] (2*\x-\lm,-\lm) to (\x-\lm,-\lm) to (\x-\lm+\m,-\lm+\m) to (2*\x-\lm-\m,-\lm+\m) to cycle;
      \end{scope}
    }

    \barpecetaro{0}{0};
    \rbarpecetaro{\lx-\llm}{\llm};
    \barpecetaro{\lx+\ly}{0};

    \draw[peco,opacity=0.25] (0,0) to (\llx,0) to (\llx-\llm,\llm) to (\llm,\llm)to cycle;

    \coordinate (x) at (\llm*3.3054,0.5833*\llm);
    \coordinate (e) at (\llx/2,\llm*0.706);
    \coordinate (A) at (\llx/2,\llx/2);
    \coordinate (B) at (\llx,0);
    \coordinate (C) at (\llx/2,-\llx/2) ;

    \draw[verda,very thick] (A) to (e) to (C);

    \draw[lumo] (0,0) to node[black, above left]{$R$} (A) to (B) to
    (B) to (C) to cycle;

    \draw[fill] (e) node[below left] {$e$} circle [radius=.1];
    \draw[lumo,thin] (\lx-\llm+\x-\lm, \llm-\lm) to (\lx-\llm+\x-\lm+\x/2,\llm-\lm+\x/2) to (\lx-\llm+2*\x-\lm, \llm-\lm) to (\lx-\llm+\x-\lm+\x/2,\llm-\lm-\x/2) to cycle;
    
  \end{tikzpicture}
    \caption{Inside the smallest light rhombus containing a trapezoid
      used in the construction of $B$, the vertical centerline
      intersects $B$ only in one point, which is an eventually middle
      point.
    \label{fig-M0}}
\end{figure}

\begin{proof}
  This follows by self-similarity, see Figure~\ref{fig-M0}. In more
  detail, let $e$ be an eventually middle point and let the big
  grey trapezoid of the figure illustrate the one after which we always
  choose $\mathcal{M}$ to reach $e$. Let $R$ be the smallest rhombus with lightlike sides which contains this trapezoid. 
  By the construction, the parts of
  the green vertical center-line-segment of $R$ which are outside the big
  trapezoid is outside $B$, for checking see Figure~\ref{figure-PF}.  Similarly, by the construction, the
  parts connecting the midpoints of the corresponding bases of the big
  and small trapezoids are outside $B$, for checking see Figure~\ref{fig-M0}. The part of $B$ covered by
  the big trapezoid and the part of $B$ covered by the small trapezoid
  in the middle are similar and can be transformed into each other by a
  homogeneous dilation with center $e$. This can be seen as follows.
  For $B^0_0$, using the
    summing formula of geometric series, we get
   $e=\left(S/2,\left(1-\frac{1}{\lambda}\right)\left(1+\frac{1}{\lambda^2}+\frac{1}{\lambda^4}+\dots\right)\right)=\left(\frac{\lambda}{3-\lambda},\frac{\lambda}{\lambda+1}\right)$. The
    bottom left corner of the middle trapezoid is
    $\left(\frac{S}{\lambda}-1+\frac{S}{\lambda^2}-\frac{1}{\lambda},
    1-\frac{1}{\lambda}\right)=
    \left(\frac{\lambda-1}{\lambda}\frac{\lambda+1}{3-\lambda},
    \frac{\lambda-1}{\lambda}\right) =\frac{\lambda^2-1}{\lambda^2}\cdot
    e$. This calculation confirms that the center of the homogeneous
    dilation is indeed $e$.  
    Hence, except $e$, every point of the
  vertical center-line-segment of $R$ is outside $B$.
\end{proof}

The statement of the next lemma is illustrated in Figures~\ref{fig-EM}, \ref{fig-M}.

\begin{lemma}\label{lem-em-intersection}
  For every eventually middle point $e\in B$, there are timelike
  curves $\tau,\tau' :[0,1]\to \mathbb{R}^2$ such that
  \begin{enumerate}[label=(\alph*)]
  \item  $\tau(1)=\tau'(1)=e$, $\tau(0)=(-\ell,x)$,
    $\tau'(0)=(\ell,x)$ for some $x\in \mathbb{R}$,
  \item the ranges of $\tau$ and $\tau'$ become part of the
    vertical line through $e$ after a while, and
  \item both $\tau$ and $\tau'$ intersect $B$ only in $e$.
  \end{enumerate}
  Moreover, there
  are $r_n,t_n\in (0,1)$ tending to $1$ when $n$ tends to infinity,
  and there are curves $\lambda_n\subset \mathbb{R}^2\setminus B$
  connecting $\tau'(r_n)$ and $\tau(t_n)$ such that the Euclidean
  length of $\lambda_n$ tends to $0$ when $n$ tends to infinity.
\end{lemma}

\begin{figure}
  \begin{tikzpicture}[scale=0.2]
    \tikzstyle{lumo}=[red,dashed]
    \tikzstyle{peco}=[black!21,fill=black!21]
    \tikzstyle{verda}=[green!84!black] 

    \tikzmath{\x=4.0; \l =2.4; \lx=\l*\x; \m=(3-\l)*\x/2/\l; \lm=\l*\m; \y=\x-2*\lm; \ly=\l*\y;  \llm=\l*\lm;\llx=\l*\lx;\lly=\l*\ly;  \lllm=\l*\llm; \lllx=\l*\llx; \llly=\l*\lly;}

    \draw[peco] (0,0) to (\lx,0) to (\lx-\lm,\lm) to (\lm,\lm)to cycle;
    \draw[peco,shift={(\ly+\lx,0)}] (0,0) to (\lx,0) to (\lx-\lm,\lm) to (\lm,\lm)to cycle;
    \draw[peco,shift={(\lx-\llm,\llm)}] (0,0) to (\lx,0) to (\lx-\lm,-\lm) to (\lm,-\lm)to cycle;

    \begin{scope}[shift={(\llx-\lllm,\lllm)}]
      \draw[peco] (0,0) to (\lx,0) to (\lx-\lm,-\lm) to (\lm,-\lm)to cycle;
      \draw[peco,shift={(\ly+\lx,0)}] (0,0) to (\lx,0) to (\lx-\lm,-\lm) to (\lm,-\lm)to cycle;
      \draw[peco,shift={(\lx-\llm,-\llm)}] (0,0) to (\lx,0) to (\lx-\lm,\lm) to (\lm,\lm)to cycle;
    \end{scope}

    \begin{scope}[shift={(-\llx+\lllm,\lllm)}]
      \draw[peco] (0,0) to (\lx,0) to (\lx-\lm,-\lm) to (\lm,-\lm)to cycle;
      \draw[peco,shift={(\ly+\lx,0)}] (0,0) to (\lx,0) to (\lx-\lm,-\lm) to (\lm,-\lm)to cycle;
      \draw[peco,shift={(\lx-\llm,-\llm)}] (0,0) to (\lx,0) to (\lx-\lm,\lm) to (\lm,\lm)to cycle;
    \end{scope}

    \draw[lumo] (-\lly-\lllm,\lllm) to (-\lly-\lllm+\lllx/2,\lllm-\lllx/2) to   (\llx+\lly+\lllm,\lllm) to  (-\lly-\lllm+\lllx/2,\lllm+\lllx/2) to cycle;
    \draw[lumo] (0,0) to (\lllm,\lllm);
    \draw[lumo] (\llx,0) to (\llx-\lllm,\lllm);

    \draw[thick] (-\lx+\lllm,\lllm) to (\lllm,\lllm);
    \draw[thick] (-\lly,0) to (\llx+\lly,0);
    \draw[thick] (-\lly-\lllm+\llm,\lllm-\llm) to (\ly+\ly-\lly+\lx-\lllm+\llm,\lllm-\llm);
    \draw[thick] (\lllm-\lx+\lm,\lllm-\lm) to (\lllm-\lm,\lllm-\lm);

    \coordinate (0) at (\llx/2,-\lllx/2+\lllm);
    \coordinate (1) at  (-\lly/25,0);
    \coordinate (2) at (\lllm-\llm-4*\ly/5,\lllm-\llm);
    \coordinate (3) at (\ly/2,\lllm-\lm);
    \coordinate (4) at (\ly/2,\lllm);
    \coordinate (5) at (\llx/2,\lllx/2+\lllm);

    \draw[verda, very thick] (0) to ([shift={(0,-2)}]0) node[right,black] {$\tau(0)=(-\ell,x)$};
    \draw[verda, very thick] (0) node[left,black] {$(-S/2,x)$} to [out=90,in=-90] (1) to [out=90,in=-90] (2)  to [out=90,in=-90] (3)  (4)  to [out=90,in=-90] (5) node[right,black] {$(S/2,x)$};

    \draw[verda, very thick] (5) to ([shift={(0,2)}]5) node[left,black] {$\tau'(0)=(\ell,x)$};
    
    \draw[verda, ultra thick] (3) to (4);

    \draw[fill] ([shift={(0,-2)}]0) circle [radius=.2];
    \draw[fill] (0) circle [radius=.2];
    \draw[fill] (1) circle [radius=.2];
    \draw[fill] (2) circle [radius=.2];
    \draw[fill] (3) circle [radius=.2];
    \draw[fill] (4) circle [radius=.2];
    \draw[fill] (5) circle [radius=.2];
    \draw[fill] ([shift={(0,2)}]5) circle [radius=.2];
    
  \end{tikzpicture}
  \caption{Through every eventually middle point $e$ there are timelike
    curves $\tau$ and $\tau'$ respectively connecting $e$ to points
    $(-\ell,x)$ and $(\ell,x)$ for some $x\in\mathbb{R}$ such that these
    curves contain no other point of $B$ apart from $e$,
    cf.\ Lemma~\ref{lem-em-intersection}. \label{fig-EM} }
\end{figure}
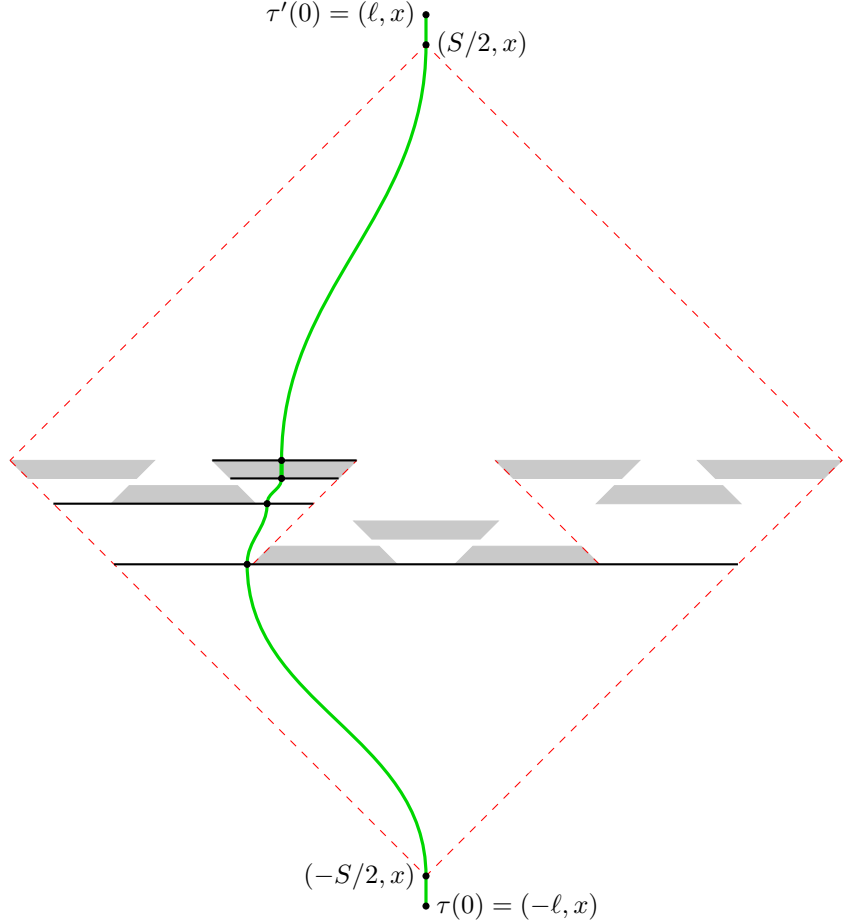

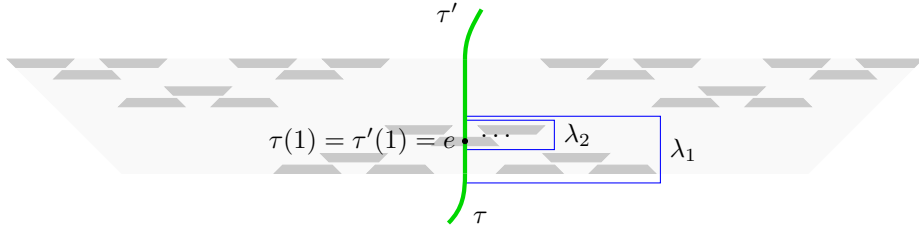
\begin{figure}
  \begin{tikzpicture}[scale=0.22]
    \tikzstyle{lumo}=[red,thick]
    \tikzstyle{peco}=[black!21,fill=black!21]
    \tikzstyle{verda}=[green!84!black] 

    \tikzmath{\x=4.0; \l =2.4; \lx=\l*\x; \m=(3-\l)*\x/2/\l; \lm=\l*\m; \y=\x-2*\lm; \ly=\l*\y;  \llm=\l*\lm;\llx=\l*\lx;\lly=\l*\ly;  \lllm=\l*\llm; \lllx=\l*\llx; \llly=\l*\lly;} 

    \newcommand{\barpecetaro}[2]{
      \begin{scope}[shift={(#1,#2)}]
        \draw[peco] (0,0) to (\x,0) to (\x-\m,\m) to (\m,\m)to cycle;
        \draw[peco,shift={(\x+\y,0)}] (0,0) to (\x,0) to (\x-\m,\m) to (\m,\m)to cycle;
        \draw[peco] (2*\x-\lm,\lm) to (\x-\lm,\lm) to (\x-\lm+\m,\lm-\m) to (2*\x-\lm-\m,\lm-\m) to cycle;
      \end{scope}
    }

    \newcommand{\rbarpecetaro}[2]{
      \begin{scope}[shift={(#1,#2)}]
        \draw[peco] (0,0) to (\x,0) to (\x-\m,-\m) to (\m,-\m)to cycle;
        \draw[peco,shift={(\x+\y,0)}] (0,0) to (\x,0) to (\x-\m,-\m) to (\m,-\m)to cycle;
        \draw[peco] (2*\x-\lm,-\lm) to (\x-\lm,-\lm) to (\x-\lm+\m,-\lm+\m) to (2*\x-\lm-\m,-\lm+\m) to cycle;
      \end{scope}
    }

    \begin{scope}[shift={(-\llx-\lly,0)}]
      \rbarpecetaro{\llx-\lllm}{\lllm};
      \barpecetaro{\llx+\lx-\lllm-\llm}{\lllm-\llm};
      \rbarpecetaro{\llx-\lllm+\lx+\ly}{\lllm};
    \end{scope}

    \barpecetaro{0}{0};
    \rbarpecetaro{\lx-\llm}{\llm};
    \barpecetaro{\lx+\ly}{0};

    \rbarpecetaro{\llx-\lllm}{\lllm};
    \barpecetaro{\llx+\lx-\lllm-\llm}{\lllm-\llm};
    \rbarpecetaro{\llx-\lllm+\lx+\ly}{\lllm};
    
    \draw[peco,opacity=0] (0,0) to (\llx,0) to (\llx-\llm,\llm) to (\llm,\llm)to cycle;
    \draw[peco,opacity=0,shift={(\llx-\lllm,\lllm)}] (0,0) to (\llx,0) to (\llx-\llm,-\llm) to (\llm,-\llm)to cycle;
    \draw[peco,opacity=0,shift={(-\llx+\lllm,\lllm)}] (0,0) to (\llx,0) to (\llx-\llm,-\llm) to (\llm,-\llm)to cycle;

    \draw[peco,opacity=0] (0,0) to (\lx,0) to (\lx-\lm,\lm) to (\lm,\lm)to cycle;
    \draw[peco,opacity=0,shift={(\ly+\lx,0)}] (0,0) to (\lx,0) to (\lx-\lm,\lm) to (\lm,\lm)to cycle;
    \draw[peco,opacity=0,shift={(\lx-\llm,\llm)}] (0,0) to (\lx,0) to (\lx-\lm,-\lm) to (\lm,-\lm)to cycle;

    \begin{scope}[shift={(\llx-\lllm,\lllm)},opacity=0]
      \draw[peco] (0,0) to (\lx,0) to (\lx-\lm,-\lm) to (\lm,-\lm)to cycle;
      \draw[peco,shift={(\ly+\lx,0)}] (0,0) to (\lx,0) to (\lx-\lm,-\lm) to (\lm,-\lm)to cycle;
      \draw[peco,shift={(\lx-\llm,-\llm)}] (0,0) to (\lx,0) to (\lx-\lm,\lm) to (\lm,\lm)to cycle;
    \end{scope}

    \begin{scope}[shift={(-\llx+\lllm,\lllm)},opacity=0]
      \draw[peco] (0,0) to (\lx,0) to (\lx-\lm,-\lm) to (\lm,-\lm)to cycle;
      \draw[peco,shift={(\ly+\lx,0)}] (0,0) to (\lx,0) to (\lx-\lm,-\lm) to (\lm,-\lm)to cycle;
      \draw[peco,shift={(\lx-\llm,-\llm)}] (0,0) to (\lx,0) to (\lx-\lm,\lm) to (\lm,\lm)to cycle;
    \end{scope}

    \draw[peco, opacity=.1, shift={(-\llx+\lllm,\lllm)}] (0,0) to (\lllx,0) to (\lllx-\lllm,-\lllm) to (\lllm,-\lllm) to cycle;


    \draw[blue]  (\llx/2,\llm*1.2) to  (\llx/2+\lx+\ly/2+\llm*0.1,\llm*1.2) to  node[black,right]{$\lambda_1$}  (\llx/2+\lx+\ly/2+\llm*0.1,-\llm*0.2) to   (\llx/2,-\llm*0.2);

    \node[] at   (\llx/2+\x/2,\llm-\lm*0.5)  {\dots};

    \draw[blue,thin]  (\llx/2,\llm*1.2-\lm*0.2) to  (\llx/2+\x+\y/2+\lm*0.5,\llm*1.2-\lm*0.2) to  node[black,right]{$\lambda_2$}  (\llx/2+\x+\y/2+\lm*0.5,\llm-\lm*1.2) to   (\llx/2,\llm-\lm*1.2);

    \draw[verda,ultra thick] (\llx/2-1,-3) to[in=-90]  node[below right,black]{$\tau$} (\llx/2,0);
    \draw[verda, ultra thick] (\llx/2,0) to (\llx/2,\lllm);

    \draw[verda,ultra thick] (\llx/2+1,\lllm+3) to[out=-120,in=90]  node[above left,black]{$\tau'$} (\llx/2,\lllm);

    \draw[fill] (\llx/2,1.95) node[left] {$\tau(1)=\tau'(1)=e$} circle [radius=.15];



  \end{tikzpicture}

  \caption{There is an infinite sequence of curves
    $\lambda_n\subset\mathbb{R}^2\setminus B$ whose Euclidean length
    tends to $0$ as $n$ tends to infinity and who ``witnesses'' in
    $\mathbb{R}^2\setminus B$ that the two halves of the centerline
    meet at $e$, cf.\ Lemma~\ref{lem-em-intersection}. \label{fig-M}}

\end{figure}

\begin{proof}
  Let $e\in B$ be an eventually middle point. There is a series of
  nested trapezoids corresponding to the choice sequence of $e$.  Let
  $B^0_n$ be the first trapezoid of this sequence. Let $x$ be such that
  $(-\ell,x)$ and $(\ell,x$) are on the vertical centerline of
  $B^0_n$. Then the edges of $B^0_n$ at its longer base are lightlike
  related to points $(-S/2,x)$ and $(S/2,x)$, cf.\ Figure~\ref{fig-EM}.

  First by induction, we show that: \emph{Each point $p$ of the lower
  bases of the trapezoids from $B^0_n$ are reachable from below from
  point $(-\ell,x)$ by a timelike curve intersecting $B$ at most this
  point $p$; and analogously, the upper bases of these trapezoids are
  reachable from above from $(\ell,x)$ by a timelike curve avoiding $B$
  in this sense.} 

  By symmetry, it is enough to show the first part of this
  statement. The base case, when $p$ is on the bottom base of $B^0_n$ is
  easy because we have that the line segment from $(-\ell,x)$ to $p$ is
  timelike since point $(-S/2,x)$ and the endpoints of the bottom base
  of $B^0_n$ are lightlike separated, and $\ell>S/2$. Since there is no
  point of $B$ below this bottom base, we can easily reach $p$ by an
  appropriate timelike curve from $(-\ell,x)$.

  To see the induction step, it is enough to observe that the bottom
  base of each trapezoid is either part of the bottom base of the
  trapezoid of the previous iteration step, and then we have already
  reached it with an appropriate timelike curve; or it is reachable from
  a point of that base which is not in $B$ by a timelike curve avoiding
  $B$; and hence, by continuing the curve given by the induction
  hypothesis with this one, we get the timelike curve we need. The above
  property used in the induction step is easy to confirm by the
  self-similarity of the construction, see Figure~\ref{fig-EM}.

  Consider now the trapezoid after which we always choose $\mathcal{M}$
  to reach $e$. By Lemma~\ref{lem-CL}, inside this trapezoid,
  the vertical centerline intersects $B$ only in the point $e$, see Figure~\ref{fig-M}.

  Now take a timelike curve from $(-\ell,x)$ to the bottom point of this
  centerline and continue it with the bottom part until $e$, with an
  appropriate parametrization this gives curve $\tau$, and a
  completely analogous way we can find a $\tau'$ connecting $(\ell,x)$
  and $e$.

  By self-similarity, to see the existence of curves $\lambda_n$, it is
  enough to see the existence of $r_1$, $t_1$ and $\lambda_1$, which is
  easy because there is plenty of space outside $B$ but inside the
  first middle trapezoid to go around the next middle trapezoid
  connecting some point $\tau(r_1)$ and $\tau'(t_1)$. Now, we can
  recursively define $\lambda_n$ by scaling down $\lambda_{n-1}$ by
  scaling factor $1/\lambda$. That the parameter points $r_n$ and $t_n$
  of $\tau$ and $\tau'$ tend to $0$, as well as, that the Euclidean
  length of $\lambda_n$ tends to $0$ when $n$ tends to infinity follows
  from the fact that, in each step, we scaled down by factor
  $1/\lambda<1/2$.
\end{proof}

\section{All extensions of the punctured time-rolled Minkowski spacetime have closed timelike curves}
\label{sec-E}

The idea of the proof is that any proper extension of $\mathbf{M^-}$ has to fill in a point of $B$ since time-rolled Minkowski spacetime is geodesically complete. Once the extension fills in a point, it fills in nearby points, of $B$, too (see the next lemma). To any point of $B$, arbitrarily close there are eventually middle points, and those are the only missing points of some CTCs in $\mathbf{M^-}$. Thus, the extension will have at least these CTCs.

We begin by proving a general lemma about extensions. We say that an extension fills in a limit for a curve living in the smaller spacetime if the curve converges to a point in the extension (but it may not converge to any point in the smaller spacetime). 
The next lemma says, intuitively, that if an extension fills in a limit, it also fills in nearby limits witnessed by short curves in a coordinate system, and two such new limits coincide if a system of short coordinate curves witness their coincidence.

Some notation: For a chart $\psi$ of $\mathbf{S}$ and broken curve $\delta:[0,1)\to S$, we denote the Euclidean coordinate-length of $\psi(\delta)$ by $|\delta|_{\psi}$. We say that $\delta$ convereges to $q$ in $S$ if there is a broken curve $\delta':[0,1]\to S$ such that $\delta'(1)=q$ and $\delta'(x)=\delta(x)$ for all $0\le x<1$. We say that $\delta$ can be continued if there is $\delta':[0,y)\to S$ with $y>1$ and $\delta'(x)=\delta(x)$ for all $0\le x<1$. 
  By $\delta\subseteq X$ we mean that the range of $\delta$ is a subset of $X$.

  The next lemma is interesting when $p\notin N$. Lemma~\ref{ext-lem}(ii) below is somewhat similar to Proposition 5.1.\ of \cite{S}.

  \begin{lemma} \label{ext-lem}
    Let $m\ge 2$, let $\mathbf{S}=(S,g)$ be an $m$-dimensonal Lorentzian manifold, let $O$ be an open set of $\mathbf{S}$, and let $\psi:N\to\mathbb{R}^m$ be a chart of $\mathbf{S}$ such that the components $g_{ij}$ and $\partial_kg_{ij}$ are bounded in the range of $\psi$. Then for all $p\in O$ there is $\varepsilon\in\mathbb{R}$ such that (i) and (ii) below hold for all broken curves $\delta, \delta':[0,1)\to S$ starting at $p$ (i.e., $\delta(0)=\delta'(0)=p)$. Let $\delta^-$ denote the curve $\delta$ without its starting point, i.e., $\delta^-=\delta\setminus\{ p\}$ and similarly for other broken curves starting at $p$. 
      \begin{description}
      \item[(i)] If $\delta^-\subseteq N$ and $|\delta^-|_{\psi}<\varepsilon$ then $\delta\subset O$ and $\delta$ can be continued in $O$.
      \item[(ii)] If $\delta,\delta'$ are as in (i) and $\psi(\delta^-), \psi(\delta'^-)$ converge to the same point in $\mathbb{R}^m$ such that this is ``witnessed by a vanishing $\psi$-ladder", then they converge to the same point in $\mathbf{S}$, too. In more detail: Let $r_n, t_n\in (0,1)$ be such that they tend to $1$ when $n$ tends to infinity. Let the curves $\lambda_n\subset N$ connect $\delta(r_n)$  with $\delta'(t_n)$ such that $|\lambda_n|_{\psi}$ converges to $0$ when $n$ tends to infinity. Then $\delta$ and $\delta'$ converge to the same point in $O$.
      \end{description}
  \end{lemma}
  \begin{proof}
    Let $\mathbf{S}$, $\psi$, $O$, $p$ and $\delta$ be as in the lemma such that $\delta\setminus\{ p\}\subset N$. Let $\xi:D\to\mathbb{R}^m$ be a chart in $\mathbf{S}$ such that $p\in D$. Since the range of $\xi$ is an open set in $\mathbb{R}^m$, we may assume that the range of $\xi$ is an open ball $G$ of radius $r$ around $\xi(p)$. By taking $r$ to be small enough, we may assume that, in $\xi$, the components $g_{ij}$ of $g$ as well as the components $\partial_kg_{ij}$ of the derivatives of $g$ are bounded by a number in $G$. Let $C_g$ be a common bound for the components of $g$ and its derivatives in the coordinate systems $\xi$ and $\psi$.

    First we prove (i). 
    Assume that $\delta$ is not a subset of $D$. There is $0<a<1$ such that $\delta(a)\notin D$ but $\delta(x)\in D$ for all $0\le x<a$. Let $\gamma$ be $\delta$ till this point, i.e., $\gamma:[0,a)\to S$ such that $\gamma(x)=\delta(x)$ for $0\le x<a$. We are going to show that $\gamma$ cannot be too short, i.e., there is $\varepsilon\in\mathbb{R}$ such that $|\gamma^-|_\psi > \varepsilon$. Therefore, $\delta\subset D\subseteq O$ if $|\delta^-|_\psi< \varepsilon$  and also $\delta$ can be continued in $D$ in this case.

      Let $k\in\mathbb{R}$ be arbitrary and assume that 
      \[ |\gamma^-|_{\psi}< k. \tag{1}\] 
      Let $q=\delta(a)$ and let $\rho$ denote the curve $\gamma$ taken from $q$ till $p$ but such that $p\notin\rho$, i.e., $\rho:[0,a)\to S$ is defined by $\rho(x)=\gamma(a-x)$ for $x\in [0,a)$. Then $|\rho|_\psi=|\gamma^-|_\psi$ and $\rho\subset N$, $\gamma\subset O$.
	  
	  From here on in the proof of (i), we will extensively rely on Section 4 of \cite{S}.
	  Take an orthonormal basis $e=(e_i : i<n)$ in $T_q\mathbf{S}$. Let $gap(\rho)$ denote the general-affine-paramater length of $\rho$ with respect to $e$.  This is denoted by $L_{gap,e}(\rho)$ in \cite{S}, its definition is recalled at the beginning of Section 4 of \cite{S}. By Lemma 4.2 of \cite{S}, there is $0<b<\omega$, depending only on $k, e$ and $C_g$, such that 
	  \[ gap(\rho)< b\cdot|\rho|_\psi \qquad\text{if}\quad |\rho|_\psi< k .\]
	  We have $|\rho|_\psi < k$ by $|\gamma^-|_{\psi}< k$.
	  
	  Let now $f=(f_i : i<n)$ in $T_p\mathbf{S}$ be the orthonormal basis which we  get if we parallel transport $e$ along $\delta$ from $q$ to $p$.  By the definition of gap-length, then $gap(\gamma)=L_{gap,f}(\gamma)$ taken with this basis $f$ is the same as $gap(\rho)=L_{gap,e}(\rho)$. Now we can apply Corollary 4.12 of \cite{S} to the chart $\xi$, since $\gamma\subset D$. It says that there are $c,d\in\mathbb{R}$ depending only on $C_g$ and $f$ such that if $gap(\gamma)< c$ then $|\gamma|_\xi < d\cdot gap(\gamma)$.  We got
	  \[|\gamma|_\xi< d\cdot gap(\gamma) = d\cdot gap(\rho)< d\cdot b\cdot |\rho|_\psi = d\cdot b\cdot |\gamma^-|_\psi, \]
	  i.e., 
	  \[ |\gamma|_\xi < d\cdot b\cdot |\gamma^-|_\psi \tag{2}\] whenever $gap(\gamma)<c$. The latter holds if
	  \[ |\gamma^-|_\psi < c\cdot b^{-1} .\tag{3}\]
	  
	  Now we use that the range of $\xi$ is a ball with radius $r$. This implies that $|\gamma|_\xi$ cannot be shorter than $r$ because it is a curve starting at the center of the ball  and leaving the ball. Hence $|\gamma|_\xi\ge r $
	  and so \[ |\gamma^-|_\psi\ge r\cdot (d\cdot b)^{-1}\tag{4}.\] Let $\varepsilon = min\{k, c\cdot b^{-1}, r\cdot (d\cdot b)^{-1} \}$. Taking this $\varepsilon$ makes (i) true. We also got in the proof that there is a bound $K\in\mathbb{R}$ such that 
	  \[ |\delta|_\xi < K\cdot |\delta^-|_\psi \tag{5}\] whenever $|\delta^-|_\psi<\varepsilon$ (namely, by (2), we can take $K=b\cdot d$).
	  
	  Proof of (ii): Let $\delta, \delta'$ and $\lambda_n$ be as in the statement of (ii). By  (i), both $\delta$ and $\delta'$ converge to points of $O$, say to $q$ and $q'$. These $q$ and $q'$ may not belong to $N$.  We want to show $q=q'$. Let the broken curves $\gamma_n$ be defined as $\delta$ from $p$ till $\delta(r_n)$ and then continued with $\lambda_n$ till $\delta'(t_n)$.  Then by our conditions, for large enough $n$, the broken curve $\gamma_n$ satisfies the conditions for (i), i.e., it starts at $p$, $\gamma_n^-\subset N$, and $|\gamma_n^-|_\psi < \varepsilon$. Thus there is $n_0$ such that $\gamma_n\subset O$, in particular $\lambda_n\subset O$ for all $n\ge n_0$. 
	  
	  Let us consider now curves starting at $q$. We have $q\in O$, by (i). Let $\varepsilon_0\in\mathbb{R}$ be the bound that exists for $q$ according to (i). Let the broken curves $\rho_n$ be defined as starting from $q$ then going in reverse direction along $\delta$ till $\delta(r_n)$, continuing along $\lambda_n$ till $\delta'(t_n)$ and then continuing along $\delta'$ till its end (so that $q'\notin\rho_n$). Then $\rho_n^-\subset N$ and $|\rho_n^-|_\psi$ tends to $0$ as $n$ tends to infinity. Let $n_1\ge n_0$ be such that  $|\rho_n^-|_\psi <\varepsilon_0$ for all $n\ge n_1$. 
	  
	  Let $K_0$ be the bound that exists for $q$ by the proof of (i), i.e., we have $|\rho_n|_\xi < K_0\cdot|\rho_n^-|_\psi$ for all $n\ge n_1$. Thus $|\rho_n|_\xi$ tends to $0$ as $n$ tends to infinity. Since $\xi(\rho_n)$ starts at $\xi(q)$ and converges to $\xi(q')$ for all $n\ge n_1$, this means  that $\xi(q)=\xi(q')$.  Hence $q=q'$ since $\xi$ is a bijection, and we are done.	
  \end{proof}

  We are ready for proving the main property of $\mathbf{M^-}$, namely,
  that it is maximal among the spacetimes that do not contain CTCs.

  \begin{prop}\label{prop-max}
    Each proper extension of $\mathbf{M^-}$ contains closed timelike curves.
  \end{prop}

  \begin{proof}  
   Assume that
    $\mathbf{S}=(S,g)$ is a proper extension of
    $\mathbf{M^-}=(M^-,\eta)$.  We may assume that $\mathbf{M^-}$ is the
    restriction of $\mathbf{S}$ to $M^-\subset S$ and $M^-$ is an open
    set in $\mathbf{S}$. 

    From now on, we will work in $\mathbf{S}$. 
    \begin{description}
    \item[Step 1] We show%
    \footnote{We give a reason, but this is well-known, see, e.g., Lemma A.6 in \cite{S}.}
     that there is geodesic $\gamma \subset M^-$
      converging to some $p\in S\setminus M^-$.
    \end{description}

    Let $q\in M^-$ and $p\in S\setminus M^-$ be arbitrary, there are
    such points. There is a broken geodesic $\gamma$ that connects
    them in $\mathbf{S}$ because $\mathbf{S}$ is connected. By the
    properties of the real numbers, and because $p\in M^-$ and
    $q\notin M^-$, there is a first point in $\gamma$ that is not in
    $M^-$. Thus, we may assume that $\gamma\subset M^-$ is a geodesic
    that converges to $p\notin M^-$ (in $\mathbf{S}$ of course), by
    letting $\gamma$ be the last portion of the broken geodesic that
    lies in $M^-$ and taking $p$ to be the first point that is not in
    $M^-$. We may assume that $\gamma\subset\mathbb{R}^2$.

    \begin{description}
    \item[Step 2] We show that $\gamma$ converges, in $\mathbb{R}^2$,
      to a point $p'\in B$.
    \end{description}

    This is so because $\gamma$ has to converge to some point we left
    out from rolled-up Minkowski spacetime as the latter is
    geodesically complete. 

    \smallskip
    We are ready to apply Lemma~\ref{ext-lem}. Let
    $N=((-1,2)\times\mathbb{R})\setminus B$. Then $N$ is an open
    subset of $M^-$ as well as it is a subset of $\mathbb{R}^2$.  Let
    $\psi:N\to\mathbb{R}^2$ be defined to be the identity, i.e.,
    $\psi(x,y)=(x,y)$ for all $(x,y)\in N$. Then $\psi$ is a chart of
    $\mathbf{M^-}$, so it is a chart of $\mathbf{S}$, too, because
    $\mathbf{S}$ is an extension of $\mathbf{M^-}$.

    We are going to apply Lemma~\ref{ext-lem} to $\mathbf{S}$ and
    $\psi$ by taking $O$ to be any open neighborhood of $p$ (in
    $\mathbf{S}$, of course).  By using this scenario and the
    properties of our concrete example $\mathbf{M^-}$, we are going to
    show that an eventually middle point in $B$ is ``filled in'' in $O$, and this
    will bring in the CTC whose only missing point in $\mathbf{M^-}$
    was this middle point.

    The conditions of Lemma~\ref{ext-lem} hold because on chart $\psi$
    the components $g_{ij}$ are $0$ or $1$ and all of $\partial
    g_{ij}=0$ since here the metric is the standard Minkowski
    metric. Let $\varepsilon\in \mathbb{R}$ be as given by
    Lemma~\ref{ext-lem} for $p$.

    \begin{description}
    \item [Step 3] There are broken curves $\delta,\delta':[0,1) \to
      S$ starting at $p$ and short enough according to chart $\psi$
      such that apart from their starting points they stay within
      $N\cap O$, and on chart $\psi$, they both converge to some
      eventually middle point $e\in B$ from opposite vertical
      directions.
    \end{description}

    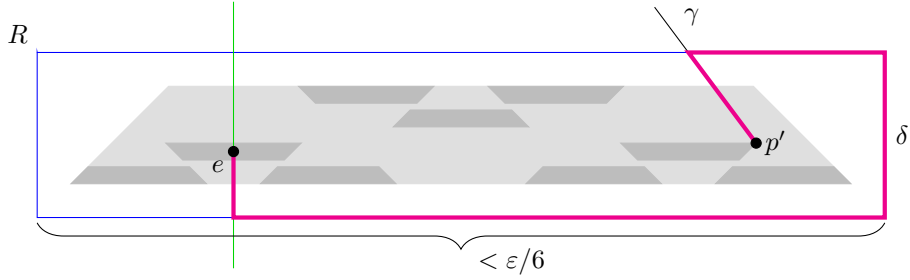
\begin{figure}[h!tb]
      \begin{tikzpicture}[scale=0.45]
        \tikzstyle{lumo}=[red,thick]
        \tikzstyle{peco}=[black!50,fill=black!50,opacity=0.35]
	\tikzstyle{verda}=[green!84!black] 

        \tikzmath{\x=4.0; \l =2.4; \lx=\l*\x; \m=(3-\l)*\x/2/\l; \lm=\l*\m; \y=\x-2*\lm; \ly=\l*\y;  \llm=\l*\lm;\llx=\l*\lx;\lly=\l*\ly;  \lllm=\l*\llm; \lllx=\l*\llx; \llly=\l*\lly;} 

        \newcommand{\barpecetaro}[2]{
          \begin{scope}[shift={(#1,#2)}]
            \draw[peco] (0,0) to (\x,0) to (\x-\m,\m) to (\m,\m)to cycle;
            \draw[peco,shift={(\x+\y,0)}] (0,0) to (\x,0) to (\x-\m,\m) to (\m,\m)to cycle;
            \draw[peco] (2*\x-\lm,\lm) to (\x-\lm,\lm) to (\x-\lm+\m,\lm-\m) to (2*\x-\lm-\m,\lm-\m) to cycle;
          \end{scope}
        }

        \newcommand{\rbarpecetaro}[2]{
          \begin{scope}[shift={(#1,#2)}]
            \draw[peco] (0,0) to (\x,0) to (\x-\m,-\m) to (\m,-\m)to cycle;
            \draw[peco,shift={(\x+\y,0)}] (0,0) to (\x,0) to (\x-\m,-\m) to (\m,-\m)to cycle;
            \draw[peco] (2*\x-\lm,-\lm) to (\x-\lm,-\lm) to (\x-\lm+\m,-\lm+\m) to (2*\x-\lm-\m,-\lm+\m) to cycle;
          \end{scope}
        }

        \barpecetaro{0}{0};
        \rbarpecetaro{\lx-\llm}{\llm};
        \barpecetaro{\lx+\ly}{0};

        \draw[peco,opacity=0.25] (0,0) to (\llx,0) to (\llx-\llm,\llm) to (\llm,\llm)to cycle;

        \coordinate (e) at (\lx/2,.95);
        \coordinate (p) at (\llx-\x+\lm,\lm);
        \coordinate (C) at (\llx/2+\lx+\ly/2+1,\llm +1);
        \coordinate (B) at (\llx/2+\lx+\ly/2+1,-1) ;
        \coordinate (A) at (\lx/2,-1);

        \draw[blue]  (-1,\llm +1) node [black,above left] {$R$} to  (C) to  (B) to   (-1,-1) to cycle;
        \draw[verda,] (\lx/2,-2.5) to (\lx/2,\llm+2.5);

        \draw[black] (p) to node[above right,black, pos=4/5] {$\gamma$}  ([shift={(-3,1.5*(\llm-\lm+1))}]p)  ;

        \draw[magenta,ultra thick]  (p)  to ([shift={(-2,\llm-\lm+1)}]p)  to (C) to node[right,black]{$\delta$} (B) to (A) to  (e);

        \draw[fill] (e) node[below left] {$e$} circle [radius=.15];
        \draw[fill] (p) node[right] {$p'$} circle [radius=.15];

        \draw [decoration={brace,amplitude=10,raise=2},decorate] (B) to node[black,below right=3, yshift=-5] {$<\varepsilon/6$} (-1,-1);

      \end{tikzpicture}
      \caption{This figure illustrates the construction of the curves
        used in Step 3. \label{fig-delta}}
    \end{figure}

    Except from the initial ones, around every trapezoid used in the
    construction of $B$, there is a surrounding rectangle not
    intersecting $B$.  By Lemma~\ref{lem-em-dense}, there is an
    eventually middle point $e\in B$ arbitrarily close to $p'$. Let
    $e$ be so close to $p'$ that they are in a non-initial trapezoid
    whose surrounding rectangle $R$ has sides less than
    $\varepsilon/6$, see Figure~\ref{fig-delta}. Since $\gamma$ tends
    to $p'$ here, we can assume without loss of generality that
    $\gamma$ intersects $R$. Let $\delta:[0,1)\to S$ be the broken
      curve defined as follows and as illustrated by
      Figure~\ref{fig-delta}: We go from $p$ along $\gamma$ backwards
      until the intersection point of $\gamma$ and $R$. From this
      point, we go around along $R$ until the intersection point of
      the vertical line through $e$ and $R$. Finally, we go along the
      vertical line until $e$. Analogously, we get $\delta'$ going to
      $e$ in the opposite direction after the first breaking point. The
      coordinate lengths $|\delta|_\psi$ and $|\delta'|_\psi$ are both
      less than $\varepsilon$ because the coordinate length of the
      part along $\gamma$ is shorter than the diagonal of $R$, which
      is less than $2\cdot \varepsilon/6$, and the rest is less than
      $4\cdot \varepsilon/6$ because it contains at most 4 segments
      each of which is shorter than the longest sides of $R$. By their
      construction and Lemma~\ref{lem-CL}, both $\delta$ and $\delta'$
      stay within $N$ apart from point
      $p=\delta(0)=\delta'(0)$. Hence, by (i) of Lemma~\ref{ext-lem},
      we have that, apart from $p$, they are also in $O$.

      \begin{description}
      \item [Step 4] In $\mathbf{S}$, $\delta,\delta':[0,1) \to S$
        converge and have the same limit point.
      \end{description}

      By (i) of Lemma~\ref{ext-lem}, $\delta$ and $\delta'$ can even
      be continued in $O$, and hence they have limit points in
      $\mathbf{S}$.  After a while, the ranges of curves $\delta$ and
      $\delta'$ overlap with those of timelike curves $\tau$ and
      $\tau'$ given by Lemma~\ref{lem-em-intersection}. Hence, by
      Lemma~\ref{lem-em-intersection}, we have the ``vanishing
      $\psi$-ladder'' witnessing that $\delta$ and $\delta'$ converge
      to the same point in chart $\psi$ required to apply (ii) of
      Lemma~\ref{ext-lem} and to conclude that they converge to the
      same point in $O$, and hence in $\mathbf{S}$.

      \begin{description}
      \item [Step 5] Through this common limit point of $\delta$ and
        $\delta'$, there is a CTC in $\mathbf{S}$.
      \end{description}

      Since after a while the ranges of curves $\delta$ and $\delta'$
      overlap with those of $\tau$ and $\tau'$ given by
      Lemma~\ref{lem-em-intersection}, the common limit point of
      $\delta$ and $\delta'$ is also a common limit point for timelike
      curves $\tau$ and $\tau'$. Hence,  going
      forwards in $\tau$ and backwards in $\tau'$ gives us the desired
      CTC in $\mathbf{S}$, because the starting points of
      $\tau$ and $\tau'$ are glued together in $\mathbf{M}^-$.
 \end{proof}

 \section{Higher dimensions}\label{sec-h}

The barrier $B\times \mathbb{R}^{d-2}$ works in $d$-dimension in an
analogous construction. We have that $B\times \mathbb{R}^{d-2}$ is closed
and its complement is connected because the direct product of closed
sets is closed and the direct product of connected sets is
connected. Hence removing $B\times \mathbb{R}^{d-2}$ from
$d$-dimensional time-rolled Minkowski spacetime gives an analogous
$d$-dimensional punctured time-rolled Minkowski spacetime
$\mathbf{M}^-_d$.

The set $B\times \mathbb{R}^{d-2}$ can also be constructed analogously
to $B$ by intersecting $d$-dimensional closed bars that have
trapezoids as $2$-dimensional cross sections. Hence that no broken
future-directed causal curve can cross $[0,1]\times \mathbb{R}^{d-1}$
without intersecting $B\times \mathbb{R}^{d-2}$ can be proven the same
way as Proposition~\ref{prop-blocking}. The only difference is that,
instead of a point, the intersection of the corresponding nested bars
give a horizontal $d-1$-dimensional subspace\footnote{Line, plane,
etc.\ depending on the dimension $d$.}, which is contained in $B\times
\mathbb{R}^{d-2}$ and crosses the causal curve trying to go through
region $[0,1]\times\mathbb{R}^{d-2}$.

Let us note that if we intersect $B\times \mathbb{R}^{d-2}$ with any plane
parallel to the plane $\{(t,x,0,\dots,0)\in \mathbb{R}^d: t,x\in
\mathbb{R}\}$, we get back $B$ in these $2$-dimensional vertical
cross sections.  Hence, even though in this higher dimensional
construction instead of eventually middle points we have eventually
middle $d-1$-dimensional subspaces, every point of these eventually
middle subspaces is an eventually middle point of some $B$ from some
$2$-dimensional vertical cross section, and these eventually middle
points also satisfy Lemmas~\ref{lem-em-dense}, \ref{lem-CL} and
\ref{lem-em-intersection} generalized to $\mathbb{R}^d$ replacing $B$
with $B\times \mathbb{R}^{d-2}$.

We have that $\mathbf{M}^-_d$ is maximal among spacetimes that do
not contain CTCs because each step of the proof of
Proposition~\ref{prop-max} goes through in this modified
construction. The same way as we did in Step 1 and Step 2, we can
find a point $p$ from the extension corresponding to a removed
point. Because the $2$-dimensional vertical cross section of
$\mathbf{M}^-_d$ through this point $p$ is isomorphic to
$\mathbf{M}^-_2$, curves $\delta$, $\delta'$, $\lambda_n$'s, $\tau$
and $\tau'$ of the two $2$-dimensional construction exist also in
$\mathbf{M}^-_d$. So, since Lemma~\ref{ext-lem} works in any
dimension, we can repeat the same proof with these curves and find the
CTC we are searching for.\\

\noindent {\bf Author Contributions} All authors wrote and reviewed the manuscript.\\

\noindent {\bf Data Availability} No datasets were generated or analysed during the current study.\\

\noindent {\Large {\bf Declarations}}\\

 \noindent {\bf  Conflict of Interest} The authors declare no conflict of interest.\\

\noindent {\bf Competing interests} The authors declare no competing interests.


%


  \medskip


  \bibliographystyle{plain}

\end{document}